\documentclass[11pt]{article}
\usepackage{graphicx}
\usepackage{pgfplots}

\usepackage[colorlinks=true,citecolor=blue]{hyperref}
\usepackage{natbib}
\usepackage{graphicx}
\usepackage{amsfonts}
\usepackage{amsmath}
\usepackage{amssymb}
\usepackage{url}
\usepackage{fancyhdr}
\usepackage{indentfirst}
\usepackage{enumerate}
\usepackage{titlesec}
\usepackage{amsthm}
\usepackage{dsfont}
\usepackage[misc]{ifsym}

\usepackage[final]{changes}

\theoremstyle{definition}

\newtheorem{definition}{Definition}

\theoremstyle{plain}
\newtheorem{theorem}{Theorem}
\newtheorem{lemma}{Lemma}
\newtheorem{proposition}{Proposition}

\theoremstyle{remark}
\newtheorem{remark}{Remark}

\usepackage{cases}

\theoremstyle{definition}

\def\N{\mathbb{N}}

\def\E{\mathbb{E}}

\def\R{\mathbb{R}}

\def\X{\mathcal{X}}

\def\X{\mathcal{X}}

\def\d{\mathrm{d}}

\usepackage[onehalfspacing]{setspace}

\usepackage{bm}
\usepackage{tikz-qtree}
\usepackage{tikz}

\def\id{\mathds{1}}

\title{Prudence and higher-order risk attitudes in the rank-dependent utility model}
\author{
	Ruodu Wang\thanks{Department of Statistics and Actuarial Science, University of Waterloo,  Canada. \Letter~{\scriptsize\url{wang@uwaterloo.ca}}}
	\and
	Qinyu Wu\thanks{Department of Statistics and Actuarial Science, University of Waterloo,  Canada. \Letter~{\scriptsize\url{q35wu@uwaterloo.ca}}}
}
\date{\today}

\pgfplotsset{compat=1.18}

\begin{document}
	\maketitle

\begin{abstract}
We obtain a full characterization  of consistency with respect to higher-order stochastic dominance within the rank-dependent utility model.
Different from the results in the literature, we do not assume any condition on the utility functions and the probability weighting functions, such as differentiability or continuity. 
It turns out that the level of generality that we offer  leads to  models that do not have a continuous probability weighting function and yet they satisfy prudence. In particular, the corresponding probability weighting function can only have a jump at $1$, and must be linear on $[0,1)$. 
 

\textbf{Keywords}: Stochastic dominance; expected utility model; completely monotone functions; probability weighting; discontinuity.
\end{abstract}

\section{Introduction}





Stochastic dominance is a widely used concept in  economics,  finance, and engineering for comparing different distributions of uncertain outcomes, particularly in the context of risk preferences in decision theory. Stochastic dominance is considered as a robust way of risk comparison as it allows for analysis without a specific utility function or preference model; see \cite{L15} and \cite{SS07}.

The most popular stochastic dominance rules are the first-order stochastic dominance (FSD) and the  second-order stochastic dominance (SSD).
More recently, higher-order risk attitudes,\footnote{By ``higher-order" we meant an order that is larger than $2$.}  captured by higher-order stochastic dominance, become popular concepts in decision theory; see e.g., \cite{ES06}, \cite{EST09}, \cite{CET13} and \cite{DS14}.
In this paper, we refer to higher-order risk attitudes as \emph{consistency with higher-order stochastic dominance}.
Among these attitudes, prudence (described by third-order stochastic dominance, TSD) is particularly significant as it relates to precautionary behavior, highlighting how prudence influences savings behavior when future income is uncertain, as shown by \cite{K90}. 
Through the concept of risk apportionment,
\cite{ES06} established elegant descriptions of consistency with respect to higher-order risk attitudes on preference relations. These studies emphasize the importance of prudence in modeling and analyzing economic behavior, making it a crucial component in the broader analysis of risk attitudes.
Accurately characterizing higher-order risk attitudes, especially prudence, within different decision-making frameworks is thus important for a deeper understanding of behavior and decision under risk.

In this paper, our goal is to fully characterize higher-order risk attitudes in the rank-dependent utility (RDU) model, introduced by \cite{Q82}.
The RDU model is one of the most popular models for decision under risk, 
and it serves as the building block for the cumulative prospect theory of \cite{TK92}. 
The RDU models include both the expected utility (EU) model and the dual utility model of \cite{Y87} as special cases.  
Characterization of other notions of risk attitudes for RDU can be found in \cite{CKS87}, \cite{W94},   \cite{R06}, and \cite{WW24a}.
We refer to \cite{W10} for a general background on RDU and related decision models. 

All EU models that are consistent with $n$th-order stochastic dominance are precisely those with an $n$-monotone utility function. 
This follows from a result of \cite{M97} and is reported in 
Proposition \ref{prop-eq}.
 Such functions have derivatives up to degree $n-2$, but not necessarily differentiable at degree $n-1$ or $n$. 

In the RDU framework,
it is straightforward to verify 
 that, for a risk-averse decision maker (i.e., SSD consistent), the utility function must be concave, and the probability weighting function must be convex; see \cite{CKS87}, where some differentiability is assumed.
The most relevant result on RDU with higher-order consistency is Theorem 2.1 of \cite{ELS20}, which characterizes the expectation through consistency with TSD under the dual utility model, where the utility function is assumed to be the identity, and the probability weighting function is assumed differentiable up to arbitrary degree. This restriction reduces the class of potential probability weighting functions and offers technical convenience.
The result of \cite{ELS20} is that among all probability weighting functions in the dual utility model,  only the identity is consistent with TSD. 
In contrast, our main result, Theorem \ref{th-main}, shows that when differentiability is not assumed, 
there are more probability weighting functions that yield dual utility models consistent with TSD, and other higher-order risk attitudes. In particular, in the dual utility model, such probability weighting functions, except for the identity, are not continuous, but they are linear on $[0,1)$ and indexed by one parameter.
The corresponding preference model is a mixture of an EU model and a worst-case, most pessimistic, RDU model.
Our results unify existing theories and offers a clear way for evaluating preference relations that align with higher-order risk attitudes. 

The rest of paper is organized as follows. In Section \ref{sec:pre}, we introduces the necessary notations and definitions. Section \ref{sec:mr} presents the main results and discusses the characterization of other notions of risk attitudes for the RDU model found in the literature. All proofs are presented in Section \ref{sec:proof}.

\section{Preliminaries}\label{sec:pre}
Let $a,b\in\R$ with $a<b$. We assume that all random variables take values in the interval $[a, b]$, and we denote this space of random variables as $\mathcal X_{[a,b]}$.
Capital letters, such as $X$ and $Y$, are used to represent random variables, and $F$ and $G$ for distribution functions.



For $X\in\mathcal X_{[a,b]}$, we write $\E[X]$ for the expectation of $X$. 
Denote by $F_X$ and $F_X^{-1}$ the distribution function and the left-quantile function of $X$, respectively, where we have the relation that
$F_X^{-1}(s)=\inf\{x:F_X(x)\ge s\}$ for $s\in(0,1]$. We use $\max(X)$ and $\min(X)$ to represent the essential supremum and the essential infimum of $X$, i.e., $\max(X):=F_X^{-1}(1)$ and
$\min(X):=\inf\{x:F_X(x)> 0\}$.
Denote by $\delta_{\eta}$ the point-mass at $\eta\in\R$. For a real-valued function $f$, let $f_-'$ and $f'_+$ be the left and right derivative of $f$, respectively, and denote by $f^{(n)}$ the $n$th derivative for $n \in \mathbb{N}$. Whenever we use the notation $f_-'$, $f'_+$ and $f^{(n)}$, it is understood that they exist.
We recall that the left derivative of a convex or concave function always exists (see e.g., Proposition A.4 of \cite{FS16}). Denote by $[n]:=\{1,\dots,n\}$ for $n\in\N$.
In this paper, all terms like ``increasing", ``decreasing", ``convex" and ``concave" are in the weak sense.

A decision maker's preference relation $\succsim$ is a weak order\footnote{That is, for $X,Y,Z\in \X$, (a)
either $X\succsim Y$ or $Y\succsim X$; (b) $X\succsim Y$ and $Y\succsim Z$
imply $X\succsim Z$.} on $\mathcal X_{[a,b]}$,  with asymmetric part $\succ$  and symmetric part $\sim$. For $X,Y\in\mathcal X_{[a,b]}$, $X\succsim Y$ means that $X$ is at least as good
as $Y$ for the decision maker.

For a distribution function $F$, denote by $F^{[1]}=F$ and define
\begin{align*}
F^{[n]}(\eta)=\int_{-\infty}^\eta F^{[n-1]}(\xi)\d \xi,~~\eta\in\R~{\rm and}~n\ge 2.
\end{align*}
It is well-known that $F_X^{[n]}(\eta)$ is connected to the expectation of $(\eta-X)_+^n$ (see e.g., Proposition 1 of \cite{OR01}):
\begin{align}\label{eq-SDeq}
F_X^{[n+1]}(\eta)=\frac{1}{n!}\E[(\eta-X)_+^n],~~X\in\mathcal X_{[a,b]},~\eta\in\R,~n\ge 1,
\end{align}
where $x_+=\max\{0,x\}$ for $x\in\R$. 

The following outlines the definitions of $n$th-order stochastic dominance.

\begin{definition}\label{def-nSD}
For $n\in \N$, we say that $X$ dominates $Y$ in the sense of $n$th-order stochastic dominance ($n$SD), denoted by $X\ge_{n} Y$ or $F_X\ge_{n} F_Y$ if 
\begin{align*}
F_X^{[n]}(\eta)\le F_Y^{[n]}(\eta),~\forall \eta\in[a,b]~{\rm and}~F_X^{[k]}(b)\le F_Y^{[k]}(b)~{\rm for}~ k\in[n]
\end{align*}
or equivalently,
\begin{align*}
\E[(\eta-X)_+^{n-1}]\le \E[(\eta-Y)_+^{n-1}],~\forall \eta\in[a,b]~{\rm and}~\E[(b-X)^{k}]\le \E[(b-Y)^{k}]~{\rm for}~ k\in[n-1].
\end{align*}
\end{definition}

For $n\in\{1,2,3\}$, 
$n$SD corresponds to the well-known   FSD, SSD, and TSD. The partial order $\ge_{n}$ for these cases is commonly written as $\ge_{\rm FSD}$, $\ge_{\rm SSD}$ or $\ge_{\rm TSD}$. 
A direct conclusion is that $n$SD is stronger than $(n+1)$SD for $n\ge 1$, i.e., $X\ge_{n}Y$ implies $X\ge_{n+1} Y$.

We say that a preference relation $\succsim$ is \emph{consistent with $n$SD} if $X\succsim Y$ for all $X,Y\in\mathcal X_{[a,b]}$ with $X\ge_n Y$. Intuitively, $n$SD   compares two uncertain outcome, and consistency with $n$SD implies that the decision maker prefers the less risky outcome according to $n$SD. Specifically, consistency with FSD states that higher outcomes are always preferred. Consistency with SSD is related to risk aversion, defined as an aversion to mean-preserving spreads (see \cite{RS70}). Consistency with higher-order stochastic dominance accommodates decision makers who exhibit more refined preference relations, such as prudence when $n=3$ (\cite{K90}) and temperance when $n=4$ (\cite{K92}). 
Their preference descriptions are obtained by \cite{ES06}
via risk apportionment, which generalizes the idea of mean-preserving spreads. 


In some literature, higher-order stochastic dominance is applicable to unbounded random variables; see \cite{R76, F80, SS07}. Contrary to Definition \ref{def-nSD}, this version is referred to as ``unrestricted" stochastic dominance because it does not impose boundary conditions at point $b$, requiring instead that $F_X^{[n]}(\eta) \leq F_Y^{[n]}(\eta)$ for all $\eta \in \mathbb{R}$. 
The $n$SD in Definition \ref{def-nSD} is a more stringent rule than its unrestricted counterpart, and while they provide the same comparisons of random variables within $\mathcal{X}_{[a,b]}$ for $n \leq 3$, distinction emerge for $n \geq 4$; see \cite{WW24}.
Note that the more stringent a stochastic dominance rule, the weaker its consistency property tends to be, leading to stronger characterization results derived from this consistency. Therefore, we opt for the restricted stochastic dominance in Definition \ref{def-nSD} over the more lenient unrestricted version.


\section{Characterization}\label{sec:mr}

\subsection{Expected utility}

Before understanding the consistency properties in rank-dependent utility models, one needs to understand the more basic expected utility (EU) model.
First, we introduce some definitions below.
{\color{blue} A preference relation $\succsim$ admits an EU representation}
with $u:[a,b]\to\R$ if
$$X\succsim Y \iff \E[u(X)]\ge \E[u(Y)].$$
To avoid trivial cases, we assume that $u$ is nonconstant and the relevant set of $u$ is defined as
\begin{align*}
\mathcal U=\{u:\R\to\R:~u~{\rm is~nonconstant}\}.
\end{align*}

\begin{definition}[$n$-monotone functions]
Let $f:[a,b]\to\R$. For $n\ge 2$, we say that $f$ is an $n$-monotone function if $(-1)^{k-1}f^{(k)}\ge 0$ for $k\in[n-2]$ and $(-1)^{n-1}f^{(n-2)}$ is decreasing and convex, where $f^{(k)}$ is the $k$th derivative of $f$ and we assume that $f^{(0)}=f$. In particular,  $f$ is a $1$-monotone function if it is increasing.
\end{definition}

The class of $n$-monotone functions 
are useful in many fields. For instance, they fully describe all Archimedean copulas in statistics; see \cite{MN09}. 
For a mathematical treatment, see \cite{W56}.  If a function $f$ is $n$-monotone for all $n\in \N$, it is called completely monotone; this property is well studied in the mathematics literature and it is closely linked to Laplace--Stieltjes transforms; see e.g., \cite{S38}. Furthermore, \cite{W89} characterized the preference relations that admit an EU representation with all completely monotone utility functions.




Intuitively, an $n$-monotone function generalizes the notion of monotonicity beyond first-order (increasing functions) and second-order (concave functions) behavior.
The next result shows that consistency with $n$SD in the EU framework can be characterized by all $n$-monotone functions.

\begin{proposition}\label{prop-eq}
Let $n\in\N$, and
suppose that {\color{blue}a preference relation $\succsim$ admits an EU representation}
 with $u\in\mathcal U$. Then, $\succsim$ is consistent with $n$SD if and only if $u$ is an $n$-monotone function.
\end{proposition}

In \cite{F76}, the congruent set $U$ of utility functions for a stochastic dominance $\ge_{\rm SD}$ is defined as follows: $X\ge_{\rm SD} Y$ if and only if $\E[u(X)]\ge \E[u(Y)]$ for all $u\in U$. 
Note that the congruent set for a stochastic dominance is not unique. 
There are several congruent sets for $n$SD that have been extensively studied (see e.g., \cite{DE13} and
Section 4.A.7 of \cite{SS07}).
Proposition \ref{prop-eq} identifies the largest congruent set of $n$SD within the class of all utility functions.

\subsection{Rank-dependent utility}

Next, we present the definition of the \emph{rank-dependent utility} (RDU) model (\cite{Q82}). An RDU function incorporates an increasing utility function $u\in \mathcal U$ and a probability weighting function $h$ that is an element of the following set:
\begin{align*}
\mathcal H=\{h:[0,1]\to[0,1]:~h~\mbox{is~increasing,}~h(0)=0, \mbox{~and~} h(1)=1\},
\end{align*}
and it has the form
\begin{align*}
R_{u,h}(X)=\int_0^\infty h\circ \overline{F}_{u(X)}(\eta)\d \eta+\int_{-\infty}^0 (h\circ \overline{F}_{u(X)}(\eta)-1)\d \eta,
\end{align*}
where $\overline{F}=1-F$ is the survival function.
If $h$ is the identity function, then the RDU model reduces to the expected utility. On the other hand,
if $u$ is the identity, then the RDU model is the \emph{dual utility} (\cite{Y87}) that has the following definition:
\begin{align*}
I_h(X)=\int_0^\infty h\circ \overline{F}_{X}(\eta)\d \eta+\int_{-\infty}^0 (h\circ \overline{F}_{X}(\eta)-1)\d \eta,
\end{align*} 
For simplicity, we consider $R_{u,h}(F_X)$ (resp.~$I_h(F_X)$) and $R_{u,h}(X)$ (resp.~$I_h(X)$) to be identical.

{\color{blue} A preference  relation $\succsim$ admits an RDU representation}
with $u\in\mathcal U$ and $h\in\mathcal H$ if $$X\succsim Y \iff R_{u,h}(X)\ge R_{u,h}(Y).$$ 
It is straightforward to see that a preference relation that admits an RDU representation satisfies consistency with FSD. Moreover, if $u$ and $h$ are both differentiable, then
consistency with SSD holds if and only if $u$ is concave and $h$ is convex; see \cite{CKS87}. 
{\color{blue}The following theorem is our main result, which establishes the characterization of consistency with higher-order stochastic dominance in the RDU model. 

\begin{theorem}\label{th-main}
Let $n\ge 3$, and
suppose that  a preference relation admits an RDU representation
with $u\in\mathcal U$ and $h\in\mathcal H$. Then, $\succsim$ is consistent with nSD if and only if the one of the two following cases hold:
\begin{itemize}
\item[(i)] $u\in\mathcal U$ is increasing and $h(s)=\id_{\{s=1\}}$ for all $s\in[0,1]$. In this case, $R_{u,h}(X)=\min u(X)$ for all $X\in \mathcal X_{[a,b]}$.

\item[(ii)] $u$ is an $n$-monotone function and  $h(s)=\lambda s\id_{\{s<1\}}+\id_{\{s=1\}}$ for all $s\in[0,1]$ with some $\lambda\in(0,1]$. In this case, $R_{u,h}(X)=\lambda\E[u(X)]+(1-\lambda)\min u(X)$ for all $X\in\mathcal X_{[a,b]}$.
\end{itemize}
\end{theorem}
}



\begin{remark}
In Theorem \ref{th-main}, the two cases can be combined under a strict monotonicity condition. 
    Suppose that a preference relation $\succsim$ that admits an RDU model is monotone for constant, that is, $c>d$ implies $c\succ d$. Then $\succsim$ is consistent with $n$SD if and only if it can be represented by some strictly increasing $n$-monotone function $u$ and $h(s)=\lambda s\id_{\{s<1\}}+\id_{\{s=1\}}$ for all $s\in[0,1]$ with some $\lambda\in[0,1]$.
    This is because in case (i), choosing different strictly increasing functions $u$ leads to the same preference relation.
\end{remark}

{\color{blue}
Prudence, which reflects a precautionary saving behavior (\cite{K90}), corresponds to consistency with TSD. Theorem \ref{th-main} includes as a special case the characterization of prudent decision makers under the RDU framework with general utility functions and probability weighting functions. In contrast to \cite{ELS20}, who showed that only the identity weighting function is consistent with TSD under some smoothness assumptions, our result characterizes a class of preference relations that can be represented as convex combinations of an EU model and a worst-case pessimistic RDU model, offering a  nontrivial extension. This is an instance where smoothness assumptions exclude some economically meaningful models.}

As an immediate consequence of Theorem \ref{th-main}, if a preference relation $\succsim$ is represented by
$$
X\succsim Y \iff \lambda \E[u(X)] + (1-\lambda) \min v(X) \ge \lambda \E[u(Y)] + (1-\lambda) \min v(Y) ,
$$
where $u$ is $n$-monotone and $v\in\mathcal U$ is increasing, then 
$\succsim$ is consistent with $n$SD. Note that this preference relation cannot be represented by the RDU model unless $u=v$.
More generally, a preference relation represented by a mixture of several RDU functions $R_{u,h}$ that are consistent with $n$SD is again consistent with $n$SD, although it does not necessarily \textcolor{blue}{admit an RDU model representation}.

\subsection{Discussion}

Next, we discuss our characterization results  and other notions of risk attitudes for the RDU model in the literature. 
Consistency with respect to each risk attitude corresponds to a set 
 $\mathcal{M} \subseteq \mathcal{U} \times \mathcal{H}$ of pairs of  utility functions and probability weighting functions.
For many risk attitudes, the set $\mathcal M$, 
has a \emph{separable form}; that is, it imposes conditions on $u\in \mathcal U$ 
 and $h\in \mathcal H$ separately, except for the trivial case that $h_{\id}(s)=\id_{\{s=1\}}$ (in this case, $u$ does not matter). 
We write this separable form as $\mathcal U_*\times \mathcal H_*$,
which means $$\mathcal M=(\mathcal U_*\times \mathcal H_*)\cup(\mathcal U\times \{h_{\id}\}).$$ Remarkably, there are some notions  of attitude that impose a joint condition on the interplay of $u$ and $h$.  In what follows, RA stands for risk aversion.

\begin{enumerate}
\item Our Theorem \ref{th-main} demonstrates that the set $\mathcal M$ corresponding to consistency with higher-order stochastic dominance has a separable form, with  $\mathcal U_*$ being the set of $n$-monotone functions,
and the $\mathcal H_*$ being linear on $[0,1)$.

\item The set $\mathcal M$ corresponding to  consistency with SSD  has a separable form; see \cite{CKS87} and \cite{R06}. In this case, $\mathcal U_*$ is the set of increasing and concave elements of $\mathcal U$,
and $\mathcal H_*$ is the set of all convex elements of $\mathcal{H}$.
This consistency is also known as strong RA. 

\item The case of FSD is trivial as all RDU models are consistent with FSD.

\item   Probabilistic risk aversion (P-RA) in the RDU model  means quasi-convexity of
the RDU functional with respect to distribution functions; see \cite{W94}.
As shown by \cite{WW24a}, the set $\mathcal M$ corresponding to P-RA   has a separable form, in which $\mathcal U_*=\mathcal U$ and $\mathcal H_*$ is slightly larger than the set of convex probability weighting functions.

\item 
 Next, we discuss some  notions of risk attitudes whose characterization in RDU 
 leads to joint conditions on $u$ and $h$. 
 \cite{CCM05} studied the characterization of monotone RA in the RDU model.\footnote{\color{blue}
For two random variables $X,Y: \Omega\to\R$, we say that $Y$ is a monotone mean-preserving increase in risk
 of $X$ if there exists a random variable $Z: \Omega\to\R$ with $\E[Z]=0$ such that $X$ and $Z$ are comonotone and $Y$ has the same distribution function as $X+Z$, where the statement that $X$ and $Z$ are comonotone means that  $(X(\omega)-X(\omega'))(Z(\omega)-Z(\omega'))\ge 0$ for all $\omega,\omega'\in\Omega$.  A preference relation $\succsim$ exhibits monotone RA if   $X\succsim Y$ whenever $Y$ is  a monotone mean-preserving increase in risk
 of $X$. } 
They showed that, under the assumption that   $u$ is continuous and strictly increasing, and   $h$ is strictly increasing, the characterization of this consistency property is $G_u\le P_h$, where $G_u$ and $P_h$ are the \emph{index of greediness} for $u$ and \emph{the  index of pessimism} for $h$, defined as 
\begin{align}\label{eq-greedypessim}
G_u= \sup_{a \leq x_1 < x_2 \leq x_3 < x_4 \leq b} 
\frac{u(x_4) - u(x_3)}{x_4 - x_3} 
\left/ \vphantom{\frac{u(x_4) - u(x_3)}{x_4 - x_3}} \right. 
\frac{u(x_2) - u(x_1)}{x_2 - x_1};~~~
P_h= \inf_{0 < s < 1} 
\frac{1 - h(s)}{1-s} 
\left/ \vphantom{\frac{1-h(s)}{1-s}} \right. 
\frac{h(s)}{s}.
\end{align}
The condition $G_u\le P_h$ clearly illustrates the interplay between $u$ and $h$.

\item Two notions of fractional SD were introduced by \cite{MSTW17}
and  \cite{HTZ20}. 
Let us focus on the most relevant cases of fractional SD between first-order and second-order SD. 
For the notion of  \cite{MSTW17} with parameter $\gamma\in(0,1)$, denoted f-$\gamma$-SD,\footnote{\color{blue} We say that $X$ dominates $Y$ under f-$\gamma$-SD if $\E[u(X)]\ge \E[u(Y)]$ for all functions $u$ satisfying $0\le \gamma u'(y)\le u'(x)$ for all $x\le y$.} 
\cite{MW22} showed that (under continuity) 
the consistency in RDU is equivalent to $Q_h\ge \gamma G_u$, 
where $G_u$ is  given in \eqref{eq-greedypessim}
and $Q_h$ is given by \begin{align*}
Q_h= \inf_{0 \leq s_1 < s_2 \leq s_3 < s_4 \leq 1} 
\frac{h(s_4) - h(s_3)}{s_4 - s_3} 
\left/ \vphantom{\frac{h(s_4) - u(s_3)}{s_4 - s_3}} \right. 
\frac{h(s_2) - h(s_1)}{s_2 - s_1}.
\end{align*}
Hence, the set $\mathcal M$ does not have a separable form.
However, for the notion of  \cite{HTZ20} with parameter $ c\in(0,1)$,  denoted  f-$c$-SD,\footnote{\color{blue} We say that $X$ dominates $Y$ under f-$c$-SD if $\E[u(X)]\ge \E[u(Y)]$ for all functions $u$ satisfying $u'(x)>0$ and $u''(x)/u'(x)\le 1/c-1$ for all $x$.} \cite{MW22} showed that the set $\mathcal M$ has a separable form, where $\mathcal U_*$ contains $u$ with $x\mapsto u(c\log(x)/(1-c))$ being concave and $\mathcal H_*$ contains all convex elements of $\mathcal H$.

\item To the best of our knowledge,
a full characterization of weak RA\footnote{\color{blue} We say that a decision maker with a preference relation $\succsim$ exhibits weak RA if $\E[X]\succsim X$ for all random variables $X$.} in the RDU model has not been established in the literature; see \cite{C95} for some sufficient conditions. 
The existing results imply that $\mathcal M$ does not have a separable form. 

\end{enumerate}

We summarize in Table \ref{tab:1} the above cases. 
To systemically understand which notions of risk attitude leads to a separable form of $\mathcal M$ seems not clear at this point.

\begin{table}[t]
    \centering\renewcommand{\arraystretch}{1.5}
    \begin{tabular}{c|c|c} 
        {Risk attitude} & $\mathcal M$ separable? & Source \\ \hline
        {FSD}  & YES  & definition \\ \hline
        {SSD}           & YES    & \cite{CKS87}; \cite{R06}      \\ \hline
        {$n$SD}         & YES    &  Theorem \ref{th-main}       \\ \hline
        {P-RA}          & YES  & \cite{W94}; \cite{WW24a}         \\ \hline
        {M-RA}          & NO    & \cite{CCM05}       \\ \hline
        {f-$\gamma$-SD} & NO    &\cite{MW22}    \\ \hline
        {f-$c$-SD}   & YES   &\cite{MW22}        \\ \hline
        {weak RA}       & NO   &\cite{C95} (not fully characterized)        \\ \hline
    \end{tabular}
    \caption{Summary of whether $\mathcal M$ has a separable form;  some results imposed regularity conditions.}
    \label{tab:1}
\end{table}


\section{Proofs}\label{sec:proof}

\subsection{Proposition \ref{prop-eq}}
Define 
\begin{align*}
\mathcal{U}_n &= \{  u  \mid u(x) = (\eta-x)_+^{n-1}, \eta \in [a,b] \} 
\cup \{  u  \mid u(x) = (b-x)^k, k \in [n-1]\}.
\end{align*}
By Definition \ref{def-nSD}, we know that $X\ge_n Y$ if and only if $\E[u(X)]\ge \E[u(Y)]$ for all $u\in\mathcal U_n$. Note that each $u\in\mathcal U_n$ is an $n$-monotone function. Moreover, the set of all $n$-monotone functions is a convex cone and closed with respect to pointwise convergence. Hence, the result follows immediately from Corollary 3.8 of \cite{M97}. \qed

\subsection{Theorem \ref{th-main}}
In this proof, we will encounter simple random variables. An explicit representation of $R_{h,u}$ with simple random variables is given below.
For $X\in\mathcal X_{[a,b]}$ with the distribution $F_X=\sum_{i=1}^n p_i\delta_{x_i}$ where $x_1\ge  \dots\ge x_n$, $p_1,\dots,p_n\ge 0$ and $\sum_{i=1}^n p_i=1$, it holds that 
$$ 
R_{u,h}(X)=\sum_{i=1}^n (h(p_1+\cdots+p_i)-h(p_1+\cdots+p_{i-1}))u(x_i).
$$

The necessity statement of Theorem \ref{th-main} is the most challenging.
To establish it, we introduce a useful lemma that outlines the constraints on $h$.

\begin{lemma}\label{lm-consh}{\color{blue}
Let $n\ge 3$, and let $u\in\mathcal U$ and $h\in\mathcal H$. If the RDU function $R_{u,h}:\mathcal X_{[a,b]}\to\R$ is consistent with nSD, then $h(s)=\lambda s\id_{\{s<1\}}+\id_{\{s=1\}}$ for all $s\in[0,1]$ with some $\lambda\in[0,1]$.}
\end{lemma}

\begin{proof}[Proof of Lemma \ref{lm-consh}]
{\color{blue}Note that a larger $n$ corresponds to a weaker form of $n$SD, which in turn corresponds to a stronger consistency condition. Therefore, it suffices to verify the result for the case of TSD, i.e., $n=3$.}
Suppose that $R_{u,h}:\mathcal X_{[a,b]}\to\R$ is consistent with TSD. Note that consistency with TSD is stronger than consistency with SSD. By Corollary 12 of \cite{R06}, we know that one of the following cases holds: (a) $u$ is increasing and concave, and $h$ is continuous and convex; (b) $u\in\mathcal U$ and $h(s)=\lambda s\id_{\{s<1\}}+\id_{\{s=1\}}$ for all $s\in[0,1]$ with some $\lambda\in[0,1)$. Case (b) is included in this lemma. 
Suppose now that Case (a) holds. We aim to show that $h$ is an identity function on $[0,1]$. Note that $u\in\mathcal U$ is increasing and concave.
We assume without loss of generality that $a<0<b$ and $u'(0)>0$ where $u'$ is the derivative of $u$.
For $n\in\N$, $\alpha\in(0,1)$ and $y,z,\epsilon\in\R$ with $0<y\le \epsilon$ and $a\le z<-n(n-1)\epsilon$ and $y+n\epsilon\le b$, we construct two simple random variables as follows:

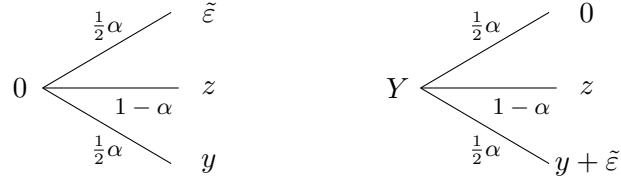
\begin{center}
\begin{tikzpicture}
  \node at (0,0) {$0$};
  \draw (0.3,0) -- (2,1) node[midway, above] {\footnotesize  $\frac{1}{2}\alpha$};
  \draw (0.3,0) -- (2,-1) node[midway, below] {\footnotesize $\frac{1}{2}\alpha$};
 \draw (0.3,0) -- (2.1,0) node[midway, below] {\footnotesize  ~~~~~~~~$1-\alpha$};
  \node at (2.5,1) {$ \tilde{\varepsilon}$};
  \node at (2.5,-1) {$y$};
    \node at (2.5,0) {$z$};
  
  \node at (5,0) {$Y$};
  \draw (5.3,0) -- (7,1) node[midway, above] {\footnotesize $\frac{1}{2}\alpha$};
  \draw (5.3,0) -- (7,-1) node[midway, below] {\footnotesize $\frac{1}{2}\alpha$};
   \draw (5.3,0) -- (7.1,0) node[midway, below] {\footnotesize ~~~~~~~~$1-\alpha$};
  \node at (7.5,1) {$0$};
  \node at (7.5,-1) {$y + \tilde{\varepsilon}$};
\node at (7.5,0) {$z$};
\end{tikzpicture}
\end{center}
where $\widetilde{\epsilon}$ is a two-point random variable with a zero mean and the distribution of the form $F_{\widetilde{\epsilon}}=(1-1/n)\delta_{n\epsilon}+(1/n)\delta_{-n(n-1)\epsilon}$. Specifically, we have
\begin{align*}
F_{X}=\left(\frac{1}{2}-\frac{1}{2n}\right)\alpha\delta_{n\epsilon}+\frac{1}{2n}\alpha\delta_{-n(n-1)\epsilon}+\frac{1}{2}\alpha\delta_y+(1-\alpha)\delta_z
\end{align*}
and
\begin{align*}
F_{Y}=\left(\frac{1}{2}-\frac{1}{2n}\right)\alpha\delta_{y+n\epsilon}+\frac{1}{2n}\alpha\delta_{y-n(n-1)\epsilon}+\frac{1}{2}\alpha\delta_0+(1-\alpha)\delta_z.
\end{align*}
Note that $y-\epsilon\le 0<y$ and $z<-n(n-1)\epsilon$, and we have
\begin{align*}
b\ge y+n\epsilon>n\epsilon>y>0>y-n(n-1)\epsilon>-n(n-1)\epsilon>z\ge a,
\end{align*} 
which implies $X,Y\in\mathcal X_{[a,b]}$. It is straightforward to check that $X\le_{\rm TSD}Y$ (see e.g., \cite{CET13}).
Denote by $a_n=(1-1/n)/2$ and $b_n=1-1/(2n)$. It holds that
\begin{align*}
R_{u,h}(X)&=u(n\epsilon)h(\alpha a_n)+u(y)(h(\alpha b_n)-h(\alpha a_n))\\
&+u(-n(n-1)\epsilon)(h(\alpha)-h(\alpha b_n))+u(z)(1-h(\alpha))
\end{align*}
and
\begin{align*}
R_{u,h}(Y)&=u(y+n\epsilon)h(\alpha a_n)+u(0)(h(\alpha b_n)-h(\alpha a_n))\\
&+u(y-n(n-1)\epsilon)(h(\alpha)-h(\alpha b_n))+u(z)(1-h(\alpha)).
\end{align*}
Since $R_{u,h}$ is consistent with TSD, we have $R_{u,h}(X)\le R_{u,h}(Y)$, and hence,
\begin{align*}
&(u(y+n\epsilon)-u(n\epsilon))h(\alpha a_n)+(u(y-n(n-1)\epsilon)-u(-n(n-1)\epsilon))(h(\alpha)-h(\alpha b_n))\\
&\ge (u(y)-u(0))(h(\alpha b_n)-h(\alpha a_n)).
\end{align*}
Letting $y\downarrow 0$, it holds that for all $\alpha\in(0,1)$, $n\ge 1$ and sufficiently small $\epsilon>0$,
\begin{align*}
u'_+(n\epsilon)h(\alpha a_n)+u_+'(-n(n-1)\epsilon)(h(\alpha)-h(\alpha b_n))\ge u'(0)(h(\alpha b_n)-h(\alpha a_n)),
\end{align*}
where $u'_+$ is the right-derivative of $u$. Letting $\epsilon\downarrow 0$ in the above equation and noting that $u'(0)>0$, we have
\begin{align*}
h(\alpha a_n)+h(\alpha)-h(\alpha b_n)\ge h(\alpha b_n)-h(\alpha a_n).
\end{align*}
With the relation that $b_n-a_n=1/2$ for all $n\ge 1$, it follows that
\begin{align*}
\frac{h(\alpha)}{\alpha}\ge \frac{h(\alpha b_n)-h(\alpha a_n)}{\alpha(b_n-a_n)},~~\forall \alpha\in(0,1),~n\ge 1.
\end{align*}
Next, letting $n\to\infty$ and noting that $h$ is continuous in Case (a) yields
\begin{align*}
\frac{h(\alpha)}{\alpha}\ge 
\frac{h(\alpha)-h(\alpha/2)}{\alpha/2},~~\forall \alpha\in(0,1).
\end{align*}
On the other hand, using the convexity of $h$ and $h(0)=0$ yields
\begin{align*}
\frac{h(\alpha)}{\alpha}\le \frac{h(\alpha)-h(\alpha/2)}{\alpha/2},~~\forall \alpha\in(0,1).
\end{align*}
Therefore, we have concluded that $h(\alpha)=2h(\alpha/2)$ for all $\alpha\in(0,1)$. Since $h$ is continuous and convex on $[0,1]$, it is straightforward to verify that $h(s)=s$ for $s\in[0,1]$. This completes the proof.
\end{proof}

\begin{proof}[Proof of Theorem \ref{th-main}]

{\bf Sufficiency.} Case (i) is trivial because $u$ is increasing and the mapping $X\mapsto\min X$ is consistent with $n$SD.
Suppose now Case (ii) holds. Proposition \ref{prop-eq} implies that the mapping $X\mapsto\E[u(X)]$ is consistent with $n$SD. Combining with the result in Case (i), we have that $R_{u,h}$ is consistent with $n$SD.

{\bf Necessity.}
By Lemma \ref{lm-consh}, we know that $h(s)=\lambda s\id_{\{s<1\}}+\id_{\{s=1\}}$ for all $s\in[0,1]$ with some $\lambda\in[0,1]$. Hence, 
\begin{align*}
R_{u,h}(X)=\lambda \E[u(X)]+(1-\lambda)\min u(X),~~X\in\mathcal X_{[a,b]}.
\end{align*}
If $\lambda=0$, then $R_{u,h}$ has the form in Case (i). If $\lambda>0$, we will verify that $X\mapsto\E[u(X)]$ is consistent with $n$SD. 
For $X,Y\in\mathcal X_{[a,b]}$ with $X\le_{n}Y$, 
define $X',Y'\in\mathcal X_{[a,b]}$ with their distributions as follows:
\begin{align*}
F_{X'}=\frac12\delta_{a}+\frac12F_{X}~~{\rm and}~~F_{Y'}=\frac12\delta_{a}+\frac12F_{Y}.
\end{align*}
It is straightforward to check that $X'\le_{n}Y'$ and $\min X'=\min Y'=a$.
Since $R_{u,h}$ is consistent with $n$SD, we have
\begin{align*}
0\le R_{u,h}(Y')-R_{u,h}(X')=\lambda(\E[u(Y')]-\E[u(X')])
=\frac{\lambda}{2}(\E[u(Y)]-\E[u(X)]).
\end{align*}
Hence, we have verified that $X\mapsto \E[u(X)]$ is consistent with $n$SD, and Proposition \ref{prop-eq}
shows that $u$ is an $n$-monotone function. This
completes the proof of necessity.
\end{proof}


\small
\begin{thebibliography}{10}

\bibitem[\protect\citeauthoryear{Chateauneuf et al.}{2005}]{CCM05}
Chateauneuf, A., Cohen, M. and Meilijson, I. (2005). More pessimism than greediness: a characterization of monotone risk aversion in the rank-dependent expected utility model. \emph{Economic Theory}, \textbf{25}(3), 649--667.


\bibitem[\protect\citeauthoryear{Chew et al.}{1987}]{CKS87}
Chew, S. H., Karni, E. and Safra, Z. (1987). Risk aversion in the theory of expected utility with rank dependent probabilities. \emph{Journal of Economic Theory}, \textbf{42}, 370--381.


\bibitem[\protect\citeauthoryear{Cohen}{Cohen}{1995}]{C95} 
Cohen, M. D. (1995). Risk-aversion concepts in expected- and non-expected-utility models. \emph{Geneva Papers on Risk and Insurance Theory}, \textbf{20}(1), 73--91.



\bibitem[\protect\citeauthoryear{Crainich et al.}{2013}]{CET13}
Crainich, D., Eeckhoudt, L. and Trannoy, A. (2013). Even (mixed) risk lovers are prudent. \emph{American Economic Review}, \textbf{103}(4), 1529--1535.

\bibitem[\protect\citeauthoryear{Deck and Schlesinger}{2014}]{DS14}
Deck, C. and Schlesinger, H. (2014). Consistency of higher order risk preferences. \emph{Econometrica}, \textbf{82}(5), 1913--1943.



\bibitem[\protect\citeauthoryear{Denuit and Eeckhoudt}{2013}]{DE13}
Denuit, M. and Eeckhoudt, L. (2013). Risk attitudes and the value of risk transformations. \emph{Journal of Economic Theory}, \textbf{9}(3), 245--254.


\bibitem[\protect\citeauthoryear{Eeckhoudt et al.}{2020}]{ELS20}
Eeckhoudt, L. R., Laeven, R. J. and Schlesinger, H. (2020). Risk apportionment: The dual story. \emph{Journal of Economic Theory}, \textbf{185}, 104971.



\bibitem[\protect\citeauthoryear{Eeckhoudt and Schlesinger}{2006}]{ES06}
Eeckhoudt, L. and Schlesinger, H. (2006). Putting risk in its proper place. \emph{American Economic Review}, \textbf{96}(1), 280--289.


\bibitem[\protect\citeauthoryear{Eeckhoudt et al.}{2009}]{EST09}
Eeckhoudt, L., Schlesinger, H. and Tsetlin, I. (2009). Apportioning of risks via stochastic dominance. \emph{Journal of Economic Theory}, \textbf{144}(3), 994--1003.






\bibitem[\protect\citeauthoryear{Fishburn}{1976}]{F76}
Fishburn, P. C. (1976). Continua of stochastic dominance relations for bounded probability distributions. \emph{Journal of Mathematical Economics}, \textbf{3}(3), 295--311.


\bibitem[\protect\citeauthoryear{Fishburn}{1980}]{F80}
Fishburn, P. C. (1980). Continua of stochastic dominance relations for unbounded probability distributions. \emph{Journal of Mathematical Economics}, \textbf{7}(3), 271--285.

\bibitem[\protect\citeauthoryear{F\"ollmer and Schied}{F\"ollmer and Schied}{2016}]{FS16} 
F\"ollmer, H. and Schied, A. (2016). \emph{Stochastic Finance. An Introduction in Discrete Time}. Fourth Edition.  {Walter de Gruyter, Berlin}.





\bibitem[\protect\citeauthoryear{Huang et al.}{2020}]{HTZ20}
Huang, R. J., Tzeng, L. Y. and Zhao, L. (2020). Fractional degree stochastic dominance. \emph{Management Science}, \textbf{66}(10), 4630--4647.



\bibitem[\protect\citeauthoryear{Kimball}{1990}]{K90}
Kimball, M. S. (1990). Precautionary saving in the small and in the large. \emph{Econometrica}, \textbf{58}(1), 53--73.



\bibitem[\protect\citeauthoryear{Kimball}{1992}]{K92}
Kimball, M. S. (1992). Precautionary motives for holding assets. In G. Hubbard (Ed.), \emph{Asymmetric Information, Corporate Finance, and Investment}. University of Chicago Press.


\bibitem[\protect\citeauthoryear{Levy}{2015}]{L15}
Levy, H. (2015). \emph{Stochastic Dominance: Investment Decision Making under Uncertainty}. Third Edition. Springer New York.

\bibitem[\protect\citeauthoryear{Mao and Wang}{2022}]{MW22}
Mao, T., and Wang, R. (2022).
	Fractional stochastic dominance in rank-dependent utility and cumulative prospect theory. \emph{Journal of Mathematical Economics}, \textbf{103}, 102766.
    

\bibitem[\protect\citeauthoryear{McNeil and Ne\v slehov\'a}{2009}]{MN09}
McNeil, A. J. and Ne\v slehov\'a, J. (2009). Multivariate Archimedean copulas, $d$-monotone functions and $\ell_1$-norm symmetric distributions.
\emph{Annals of Statistics}, \textbf{37}(5), 3059--3097.



\bibitem[\protect\citeauthoryear{M\"uller et al.}{2017}]{MSTW17}
M\"uller, A., Scarsini, M., Tsetlin, I. and Winkler. R. L. (2017).
Between first and second-order stochastic dominance. \emph{Management Science}, \textbf{63}(9), 2933--2947.






\bibitem[\protect\citeauthoryear{M\"{u}ller}{1997}]{M97}
M\"{u}ller, A. (1997). Stochastic orders generated by integrals: a unified study. \emph{Advances in Applied probability}, \textbf{29}(2), 414--428.


\bibitem[\protect\citeauthoryear{Ogryczak and Ruszczy\'{n}ski}{2001}]{OR01}
Ogryczak, W. and Ruszczy\'{n}ski, A. (2001). On consistency of stochastic dominance and mean–semideviation models. \emph{Mathematical Programming}, \textbf{89}, 217--232.







\bibitem[\protect\citeauthoryear{Quiggin}{1982}]{Q82}
Quiggin, J. (1982). A theory of anticipated utility. \emph{Journal of Economic Behavior \& Organization}, \textbf{3}(4), 323--343.




\bibitem[\protect\citeauthoryear{Rolski}{1976}]{R76}
Rolski, T. (1976). Order relations in the set of probability distribution functions and their applications in queueing theory. \emph{Dissertationes Mathematicae} CXXXII. Warsaw, Poland: Polska Akademia Nauk, Instytut Matematyczny.



\bibitem[\protect\citeauthoryear{Rothschild and Stiglitz}{1970}]{RS70}
Rothschild, M. and Stiglitz, J. (1970). Increasing risk: I. A definition. \emph{Journal of Economic Theory}, \textbf{2}(3), 225--243.


\bibitem[\protect\citeauthoryear{Ryan}{2006}]{R06}
Ryan, M. J. (2006). Risk aversion in RDEU. \emph{Journal of Mathematical Economics}, \textbf{42}(6), 675--697.


\bibitem[\protect\citeauthoryear{Shaked and Shanthikumar}{2007}]{SS07}
Shaked, M. and Shanthikumar, J. G. (2007). \emph{Stochastic Orders}. Springer New York.


\bibitem[\protect\citeauthoryear{Schoenberg}{1938}]{S38}
Schoenberg, I. J. (1938). Metric spaces and completely monotone functions. \emph{Annals of Mathematics}, \textbf{39}(4), 811--841.





\bibitem[\protect\citeauthoryear{Tversky and Kahneman}{1992}]{TK92}
Tversky, A. and Kahneman, D. (1992). Advances in prospect theory: Cumulative representation of Uncertainty. \emph{Journal of Risk and Uncertainty}, \textbf{5}(4), 297--323.





\bibitem[\protect\citeauthoryear{Wakker}{Wakker}{1994}]{W94}
 Wakker, P. (1994). Separating marginal utility and probabilistic risk aversion. \emph{Theory and Decision}, \textbf{36}(1), 1--44.
 

\bibitem[\protect\citeauthoryear{Wakker}{Wakker}{2010}]{W10}
Wakker, P. P. (2010). \emph{Prospect Theory: For Risk and Ambiguity}. Cambridge University Press.



\bibitem[\protect\citeauthoryear{Wang and Wu}{Wang and Wu}{2025a}]{WW24a}
Wang, R. and Wu, Q. (2025a).  
Probabilistic risk aversion for generalized rank-dependent functions.  \emph{Economic Theory}, \textbf{79}, 1055--1082.



\bibitem[\protect\citeauthoryear{Wang and Wu}{Wang and Wu}{2025b}]{WW24}
Wang, R. and Wu, Q. (2025b). The reference interval in higher-order stochastic dominance. 
\emph{Economic Theory Bulletin}, forthcoming.



\bibitem[\protect\citeauthoryear{Whitmore}{1989}]{W89}
Whitmore, G. A. (1989). Stochastic dominance for the class of completely monotonic utility functions. \emph{In Studies in the Economics of Uncertainty: In Honor of Josef Hadar} (pp. 77-88). New York, NY: Springer New York.


\bibitem[\protect\citeauthoryear{Williamson}{1956}]{W56}
Williamson, R. E. (1956). Multiply monotone functions and their Laplace transforms. \emph{Duke Mathematical Journal}, \textbf{23}(2), 189--207.



\bibitem[\protect\citeauthoryear{Yaari}{Yaari}{1987}]{Y87}
{Yaari, M. E.} (1987). The dual theory of choice under risk. \emph{Econometrica}, \textbf{55}(1), 95--115.

\end{thebibliography}
\end{document}